\documentclass[12pt]{extarticle}
\usepackage{amssymb}
\usepackage{amsmath}
\usepackage{amsfonts}
\usepackage{tabularx}
\usepackage{mathtools}
\usepackage{graphicx}
\usepackage{dcolumn}
\usepackage{bm}
\usepackage[mathlines]{lineno}
\usepackage[
style=numeric-comp,
sorting=none
]{biblatex}
\def \figwidth {0.8}

\usepackage{authblk}
\usepackage{amsthm}
\newtheorem{theorem}{Theorem}

\title{Time-Reversal Symmetry Bounds in Temporal Coupled-Mode Theory}
\author{Ken X. Wang
\thanks{To whom correspondence should be addressed. Electronic mail: wxz@hust.edu.cn}
\thanks{School of Physics, Huazhong University of Science and Technology, Wuhan, Hubei, China}
}

\begin{document}
\maketitle
\begin{abstract}
We provide a general treatment for the temporal coupled-mode theory with arbitrary numbers of modes and ports, and derive tight bounds for the coupling phases in addition to coupling strengths under time-reversal symmetry. We report trade-offs between the cardinalities of the reciprocal regions of the resonant coupling strengths versus phases. We discover that time-reversal symmetry enforces projected generalized reflections in the background to cancel out completely. In double-port systems, the reciprocal regions of the coupling phases span a quarter of the nonreciprocal regions for any non-hidden mode.
\end{abstract}

The temporal coupled-mode theory is a phenomenological theory for Fano-Feshbach resonances and is widely used to describe and predict physical devices~\cite{Hau83,Sie86,HH91,BD02,Fan08,FSJ03,SWF04,WF05b,JBS+02,FYLF11,ZMZ12,RF12,ZSO+13,WYSF13,FSS+14,SA17a,LHZ+23,WYL+14a,FSS+14,XC15,FSS+14,ESSA14,YCH+15,JJWM08,Fan08,SKM09,YRF10b,LAGL10,LLG+11,Pile12,RF12,ZMZ12,VYR+12,HZL+13,FSS+14,MCH14,JQW+15,SF16,AYF17,KSM18,Wang15,Dana18,DPM18,GLY+23,CTK24}. The input and output wave amplitudes can be observed at ports~\cite{SWF04}, characterizing a scattering matrix of the physical system. Waves propagate through the system via two pathways: a background pathway without interacting with localized modes, and a resonant pathway with such perturbative interaction through the coupling between the ports and the leaky modes. Waves from these two pathways coherently form the overall scattering amplitudes. These two pathways, however, are not independent but constrained by symmetries, and the coupling parameters between modes and ports cannot be arbitrary but are related to the scattering background. In particular, it was previously shown that there exists a bound on the coupling strengths in the single-mode double-port case~\cite{WYSF13}. In this Letter, we provide a general treatment for any numbers of modes and ports, and show that the symmetry constraints not only bound coupling strengths but also the phases.

A temporal coupled-mode system is shown in Figure~\ref{fig1}. We assume that there are $m$ modes and $n$ ports. In electromagnetics, for instance, a mode is a localized solution of the Maxwell's equations, and a port corresponds to a propagating solution. It is perhaps worth noting that our abstraction of the modes and ports might not intuitively match physical or geometrical shapes of the actual structures.
\begin{figure}
\begin{center}
\includegraphics[width=\figwidth\textwidth]{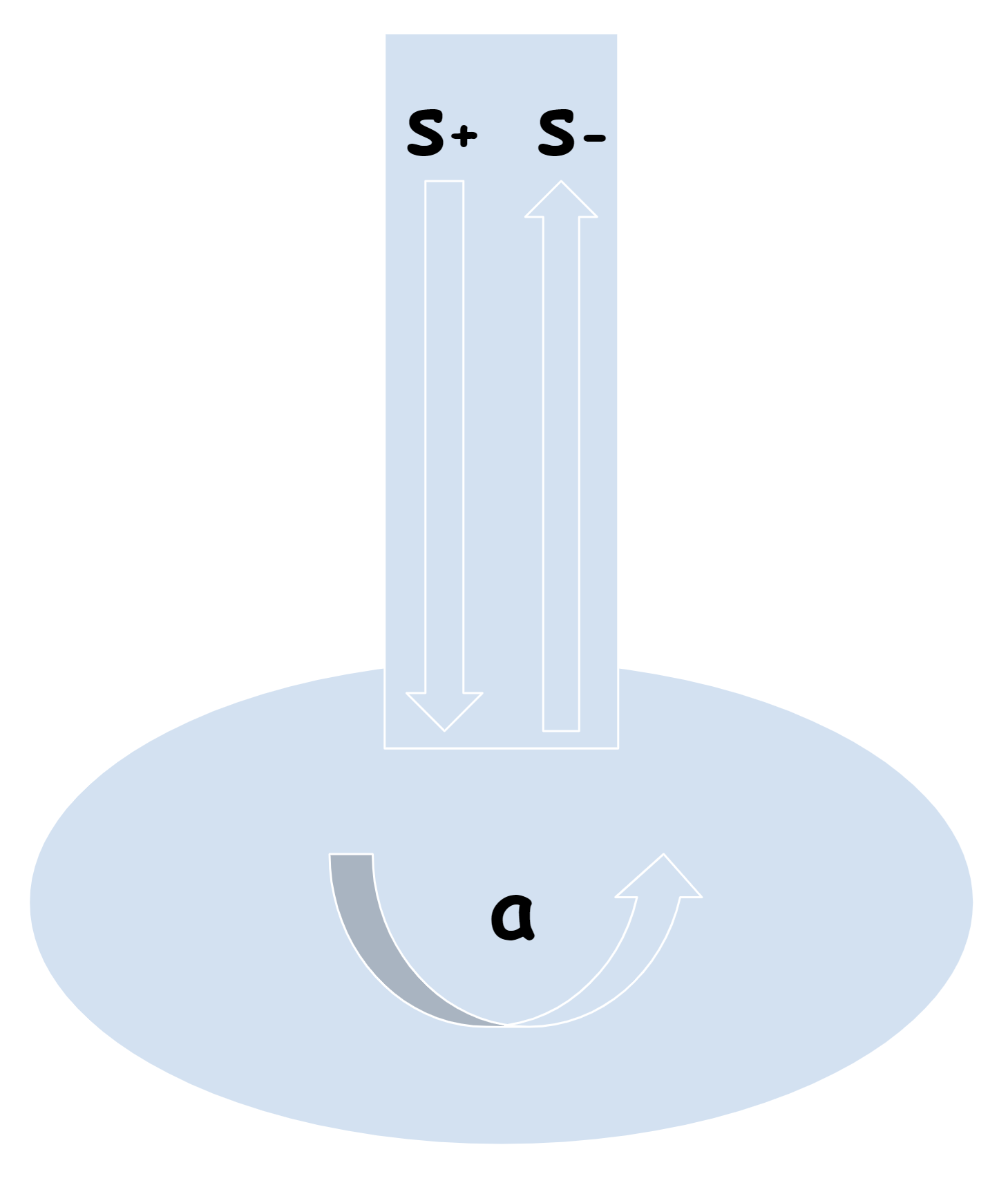}
\caption{Structure abstraction in temporal coupled-mode theory. $\mathbf{a}$ is the amplitude vector of the $m$ localized modes supported by a physical structure(s), which has many possibilities. $\mathbf{s_+}$ and $\mathbf{s_-}$ are the incoming and outgoing amplitude vectors of propagating waves through the $n$ ports, which could also correspond to a variety of different physical structures.}
\label{fig1}
\end{center}
\end{figure}
The coupled-mode equations model the system dynamics as follows:
\begin{equation}
\dot{\mathbf{a}}=(j\Omega - \Gamma)\mathbf{a}+K^T\mathbf{s_+},
\label{eq:cmt1}
\end{equation}
\begin{equation}
\mathbf{s_-}=C\mathbf{s_+}+D\mathbf{a},
\label{eq:cmt2}
\end{equation}
where $\mathbf{a}$ is the mode amplitude vector, $\mathbf{s_+}$ and $\mathbf{s_-}$ are the input and output wave amplitude vectors at the ports, $\Omega$ and $\Gamma$ are the frequency and decay matrices of the modes, $K$ and $D$ are the in-coupling and out-coupling matrices between the modes and the ports, and $C$ is the background or the direct scattering matrix between the ports without involving the modes. Equation~\ref{eq:cmt1} describes the evolution of the modes with incoming waves from the ports, and Equation~\ref{eq:cmt2} coherently combines scattering along the background and the resonant-assisted pathways.

The system parameters above are not arbitrary but obey physical laws.
\begin{table}
\begin{center}
\begin{equation} \Omega^\dagger = \Omega \end{equation} \\
\begin{equation}\label{eq:cci}C^\dagger C = I\end{equation} \\
\begin{equation}\label{eq:ddg}D^\dagger D = 2 \Gamma\end{equation}  \\
\begin{equation}\label{eq:ckd} C K^* = -D\end{equation} \\
\begin{equation}\label{eq:kdn}\vert \vert K \vert \vert = \vert \vert D \vert \vert\end{equation} 
\end{center}
\caption{Constraints on temporal coupled-mode system parameters under energy conservation.}
\label{energy}
\end{table}
If a temporal coupled-mode system obeys energy conservation alone without time-reversal symmetry, the system parameters are constrained by the equations in Table~\ref{energy}~\cite{SWF04,Wang18,ZGF19}. In particular, $\Gamma^\dagger = \Gamma$ as implied by Equation~\ref{eq:ddg}.
\begin{table}
\begin{center}
\begin{equation} \Omega^\dagger = \Omega \nonumber\end{equation} \\
\begin{equation}C^\dagger C = I\nonumber\end{equation} \\
\begin{equation}\label{eq:cct} C = C^T \end{equation} \\
\begin{equation}D^\dagger D = 2 \Gamma\nonumber\end{equation} \\
\begin{equation}C K^* = -D\nonumber\end{equation} \\
\begin{equation}\label{eq:kd}K =  D \end{equation} 
\end{center}
\caption{Constraints on temporal coupled-mode system parameters under energy conservation and time-reversal symmetry.}
\label{time}
\end{table}
With both energy conservation and time-reversal symmetry, which together imply reciprocity~\cite{Hau83,FSJ03,Ach04,JJWM08,SM22,GZF22,GF22}, the system parameters are constrained by the equations in Table~\ref{time}. In particular, the in-coupling and out-coupling matrices are equal as shown in Equation~\ref{eq:kd}.

Combining Equations~\ref{eq:ckd} and~\ref{eq:kd}, we have:
\begin{equation}
    CD^* = -D.
    \label{cdd}
\end{equation}
We observe that the background pathway, described by the background scattering matrix $C$, and the resonant pathway under reciprocity, described by the coupling matrix $D$, are not independent of each other, but related by Equation~\ref{cdd}. $C \in \mathrm{U}(n) \subset \mathbb{C}^{n \times n}$ and $D \in \mathbb{C}^{n \times m}$, in particular, we can write down each element in $D$:
\begin{equation}
\label{eqd}
D =
\left(
\begin{array}{cccc}
d_{11} & d_{12} & \cdots & d_{1m} \\
d_{21} & d_{22} & \cdots & d_{2m}\\
\vdots & \vdots & \ddots & \vdots \\
d_{n1} & d_{n2} & \cdots & d_{nm}\\
\end{array}
\right),
\end{equation}
where the matrix element $d_{pq}$ is the coupling constant between the $p$-th port and the $q$-th mode, where $p=1,2,\ldots,n$ and $q=1,2,\ldots,m$. From Equations~\ref{cdd} and~\ref{eqd}, we can derive that:
\begin{equation}
\label{eq:qq}
C\left(
\begin{array}{cccc}
d_{1q} \\
d_{2q}\\
\vdots \\
d_{nq}\\
\end{array}
\right)^* =- \left(
\begin{array}{cccc}
d_{1q} \\
d_{2q}\\
\vdots \\
d_{nq}\\
\end{array}
\right), \;\forall q.
\end{equation}
Equation~\ref{eq:qq} holds for each mode, and implies that the constraint is imposed on each mode independently. On the system level, one could engineer a mode without necessarily affecting the coupling properties of other modes. We denote $\theta_{pq} = \angle d_{pq}$ as the phase angle of the matrix element $d_{pq}$ in $D$, then $d_{pq} = \vert d_{pq} \vert e^{j \theta_{pq}}$. We can then expand Equation~\ref{eq:qq} into the following:
\begin{equation}
C
\left(
\begin{array}{cccc}
e^{-j\theta_{1q}} & 0 & \cdots & 0 \\
0 & e^{-j\theta_{2q}} & \cdots & 0\\
\vdots & \vdots & \ddots & \vdots \\
0 & 0 & \cdots & e^{-j\theta_{nq}}\\
\end{array}
\right)
\left(
\begin{array}{cccc}
\vert d_{1q} \vert \\
\vert d_{2q} \vert \\
\vdots \\
\vert d_{nq} \vert \\
\end{array}
\right) =- \left(
\begin{array}{cccc}
e^{j\theta_{1q}} & 0 & \cdots & 0 \\
0 & e^{j\theta_{2q}} & \cdots & 0\\
\vdots & \vdots & \ddots & \vdots \\
0 & 0 & \cdots & e^{j\theta_{nq}}\\
\end{array}
\right)
\left(
\begin{array}{cccc}
\vert d_{1q} \vert \\
\vert d_{2q} \vert \\
\vdots \\
\vert d_{nq} \vert \\
\end{array}
\right), \;\forall q.
\label{eq:long}
\end{equation}
Equation~\ref{eq:long} can be written as:
\begin{equation}
    G \left(
\begin{array}{cccc}
\vert d_{1q} \vert \\
\vert d_{2q} \vert \\
\vdots \\
\vert d_{nq} \vert \\
\end{array}
\right) = 0,\;\forall q,
\label{eq:gd}
\end{equation}
where
\begin{equation}
    G = C
\left(
\begin{array}{cccc}
e^{-j\theta_{1q}} & 0 & \cdots & 0 \\
0 & e^{-j\theta_{2q}} & \cdots & 0\\
\vdots & \vdots & \ddots & \vdots \\
0 & 0 & \cdots & e^{-j\theta_{nq}}\\
\end{array}
\right) + \left(
\begin{array}{cccc}
e^{j\theta_{1q}} & 0 & \cdots & 0 \\
0 & e^{j\theta_{2q}} & \cdots & 0\\
\vdots & \vdots & \ddots & \vdots \\
0 & 0 & \cdots & e^{j\theta_{nq}}\\
\end{array}
\right).
\label{eq:g}
\end{equation}

We examine all possible solutions of Equation~\ref{eq:gd}. One solution is that:
\begin{equation}
\left(
\begin{array}{cccc}
\vert d_{1q} \vert \\
\vert d_{2q} \vert \\
\vdots \\
\vert d_{nq} \vert \\
\end{array}
\right) = 0.
\label{eq:dzero}
\end{equation}
Equation~\ref{eq:dzero} describes a $q$-th mode uncoupled from all the ports. Substituting Equation~\ref{eq:dzero} into Equation~\ref{eq:ddg}, we observe that such a mode cannot leak to ports through other modes, either. Consequently, such modes could not couple to the ports either directly or indirectly, however, they could still contribute to the overall scattering process through coupling with other modes if the overlap integral with some other mode is nonzero~\cite{SWF04}. We refer to such mode as a hidden mode. In fact, if we do not impose any constraints on the coupling phases, then almost surely $\mathrm{rank}(G)=n$ and the modes are necessarily hidden. In other words, for non-hidden modes, the coupling phases are necessarily constrained.

For non-hidden modes, $G$ is singular by Equation~\ref{eq:gd}. Furthermore, if we do not impose any constraints on the coupling strengths, then we almost surely have:
\begin{equation}
    G = 0.
    \label{eq:g0}
\end{equation}
Equation~\ref{eq:g0} is only possible if $C$ is diagonal, and one must carefully engineers the coupling phases. Specifically, the background scattering matrix has to be given by:
\begin{equation}
C=
\left(
\begin{array}{cccc}
e^{j\delta_{1}} & 0 & \cdots & 0 \\
0 & e^{j\delta_{2}} & \cdots & 0\\
\vdots & \vdots & \ddots & \vdots \\
0 & 0 & \cdots & e^{j\delta_{n}}\\
\end{array}
\right),
\label{eq:cdiag}
\end{equation}
where there are no transmissions, and the reflection at the $p$-th port introduces a phase shift of $\delta_p$. With Equations~\ref{eq:g},~\ref{eq:cdiag} and~\ref{eq:g0}, we derive:
\begin{equation}
\theta_{pq} = \frac{\delta_p \pm \pi}{2},
\label{eq:theta_delta}
\end{equation}
for the $q$-th mode and the $p$-th port. It is remarkable that all mode-port couplings are necessarily in perfect synchronization among different modes. Equations~\ref{eq:cdiag} and~\ref{eq:theta_delta} establish the condition for Equation~\ref{eq:g0} to hold.

If Equations~\ref{eq:cdiag} and~\ref{eq:theta_delta} are not satisfied simultaneously, for a non-hidden mode such that Equation~\ref{eq:dzero} does not hold, we have $0 < \mathrm{rank}(G) < n$. There still exist uncountably infinite many solutions of the coupling strengths $\vert d_{pq} \vert$, but with constraints, although many of which are redundant due to the rank deficiency and complexness of $G$. The kernel of $G$ has a dimensionality of $n-\mathrm{rank}(G)$ according to the rank-nullity theorem~\cite{Mey00}. Thus the lower rank $G$ possesses, there are less restrictions on the coupling strength and more restrictions on the coupling phase. One could leverage this trade-off to engineer desirable coupling. In general, \begin{equation}\det(G)=0,\label{eq:detg0}\end{equation} and we can compute the determinant $\det(G)$~\cite{Cos06}:
\begin{equation}
\det(G) = \sum_{\sigma \subseteq [n]} \left[\mathrm{minor}_{\sigma} \left(
C
\mathrm{diag}(
e^{-j\theta_{1q}}, e^{-j\theta_{2q}} , \cdots, e^{-j\theta_{nq}}
) \right)
\mathrm{minor}_{[n] \setminus\sigma}
\mathrm{diag}(
e^{j\theta_{1q}}, e^{j\theta_{2q}} , \cdots, e^{j\theta_{nq}})\right],
\label{eq:detg}
\end{equation}
where $[n]:= \{1, 2, \cdots, n\}$, and $M_\sigma$ or $\mathrm{minor}_\sigma(M)$ is the principal minor of a square matrix $M$ indexed by $\sigma$, with $C_\emptyset = 1$.
Equations~\ref{eq:detg0} and~\ref{eq:detg} together can be simplified as:
\begin{equation}
\sum_{\sigma \subseteq [n]} \left(
C_\sigma
 e^{-j \sum_{p \in \sigma}\theta_{pq} +
j \sum_{p \in [n] \setminus\sigma} \theta_{pq}}
\right) = 0.
\label{eq:detgc}
\end{equation}
Equation~\ref{eq:detgc} has intriguing physical implications. $C_\sigma$ can be interpreted as a generalized complex reflection coefficient upon the collection of ports $\sigma$. We can always divide the power set of $[n]$ into two complementary halves $N$ and $2^{[n]}\setminus N$ such that, for any element $\sigma \in N$, $[n]\setminus\sigma \in 2^{[n]}\setminus N$. Equation~\ref{eq:detgc} can be written as:
\begin{equation}
\sum_{\sigma \in N} \left(
C_\sigma
 e^{-j \sum_{p \in \sigma}\theta_{pq} +
j \sum_{p \in [n] \setminus\sigma} \theta_{pq}} +
C_{[n] \setminus\sigma}
 e^{-j \sum_{p \in [n] \setminus\sigma}\theta_{pq} +
j \sum_{p \in \sigma} \theta_{pq}}
\right) = 0.
\label{eq:detN}
\end{equation}
Using the mathematical results on complementary principal minors of unitary matrix in the Appendix, concerning both magnitude and phase, we can further derive that:
\begin{equation}
\sum_{\sigma \in N} \vert C_\sigma \vert \cos(\eta_\sigma) = 0,
\label{eq:det_eta}
\end{equation}
where
\begin{equation}
\label{eq:eta}
\eta_\sigma = \frac{\angle(C_{[n]})}{2} -\angle(C_\sigma)  + \sum_{p \in \sigma}\theta_{pq} -
 \sum_{p \in [n] \setminus\sigma} \theta_{pq},
\end{equation}
and $\angle \cdot$ represents the phase angle. Equation~\ref{eq:det_eta} determines the constraints on the coupling phases in relation to the background process $C$. $\vert C_\sigma \vert \cos(\eta_\sigma)$ is a projection of the generalized reflection, and $2^{n-1}$ such terms cancel out in order to satisfy the time-reversal symmetry.

For single-port systems that $n=1$, the most general background is given by $C = e^{j\phi}$, and a single-port version of Equation~\ref{eq:theta_delta} can be derived from Equation~\ref{eq:det_eta}. Therefore, in a single-port reciprocal system, there are no restrictions on the coupling strength of the non-hidden modes, but the couplings need to be in perfect synchronization.

For double-port systems that $n=2$, the most general $C$ can be written as:
\begin{equation}
C = e^{j \phi} \left(
\begin{array}{cc}
r & jt \\
jt & r\\
\end{array}
\right)
\label{eq:c2}
\end{equation}
where $r^2+t^2=1$.
Substituting Equation~\ref{eq:c2} into Equation~\ref{eq:det_eta},
\begin{equation}
    \cos(\theta_{1q} + \theta_{2q} - \phi) + r \cos(\theta_{1q} - \theta_{2q}) = 0.
\end{equation}
Since $0\leqslant r \leqslant 1$, we have:
\begin{equation}
    -1 \leqslant \frac{\cos(\theta_{1q} + \theta_{2q} - \phi)}{\cos(\theta_{1q} - \theta_{2q})} \leqslant 0.
    \label{eq:phasebound}
\end{equation}
Equation~\ref{eq:phasebound} provides the bound on feasible coupling phases for any $q$-th mode under time-reversal symmetry. If Equation~\ref{eq:phasebound} is violated, then no reciprocal system is possible. Without loss of generality, we can choose appropriate reference planes such that $\phi=0$, and plot the feasible phases in Figure~\ref{fig2}. Exactly one quarter of the codomain, or the nonreciprocal feasible region, forms the feasible region with reciprocity.
\begin{figure}
\begin{center}
\includegraphics[width=\figwidth\textwidth]{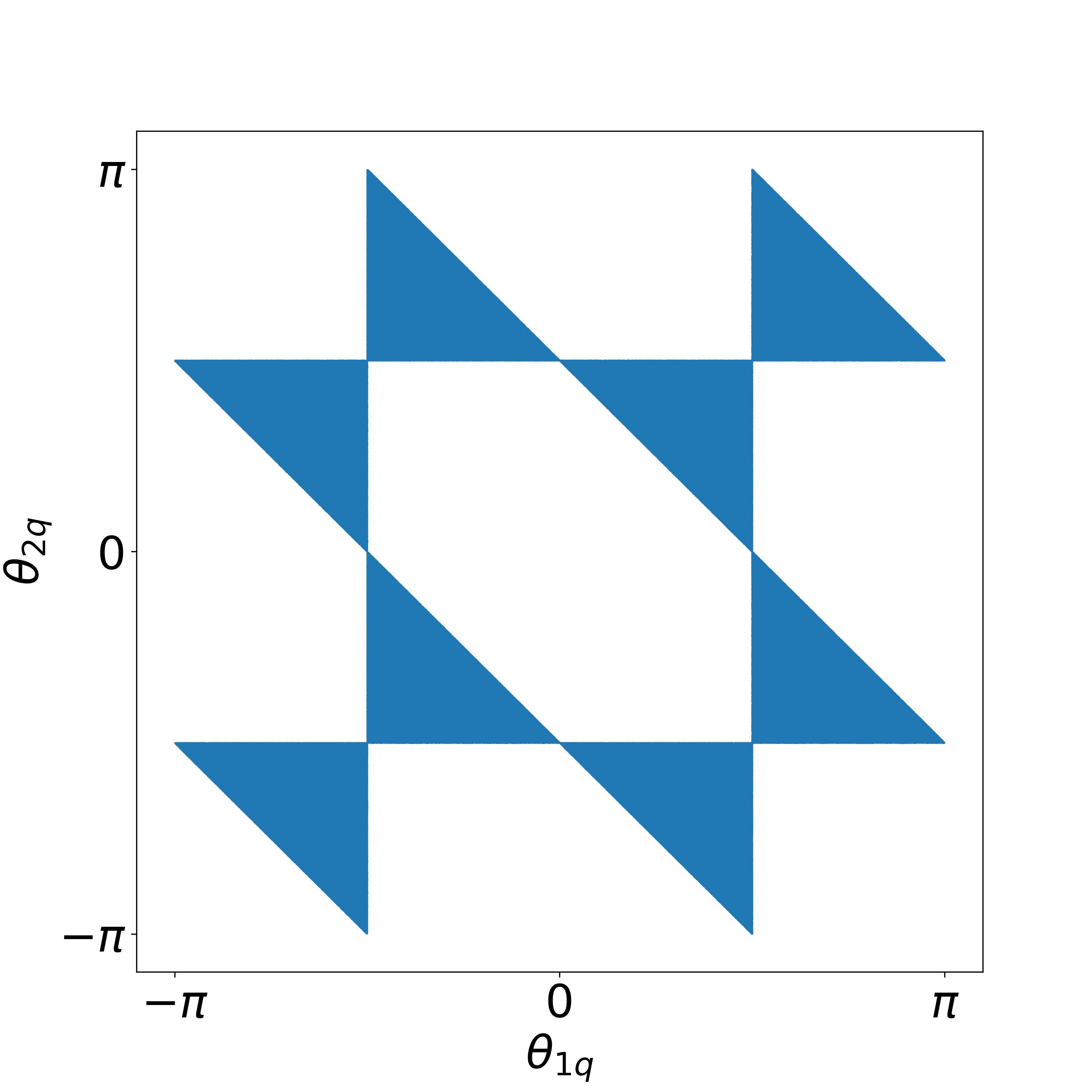}
\caption{Bounds on coupling phases under time-reversal symmetry. The colored regions are the reciprocal feasible regions of the phase angles at $n=2$, and the uncolored regions are necessarily nonreciprocal. The bounds are identical for any $q$-th non-hidden mode.}
\label{fig2}
\end{center}
\end{figure}
We can also use Equation~\ref{eq:gd} to derive that:
\begin{equation}
\label{eq:mag}
\frac{1-r}{1+r} \leqslant \left\vert \frac{d_{1q}}{d_{2q}} \right\vert^2 \leqslant \frac{1+r}{1-r}
\end{equation}
for each $q$. Equation~\ref{eq:mag} is a bound on the coupling strengths consistent with the results previously reported~\cite{WYSF13,SA17a}, and is more general since it holds for every single mode and beyond electromagnetics. Equation~\ref{eq:mag} further implies that a mode hidden from a port needs to be hidden from both ports.

When $n=3$, we can compute the bounds on coupling strengths and phases similarly, by utilizing the results on U(3) parametrization~\cite{Bro88,Jar05}. We could also obtain simpler closed-form bounds for systems with other symmetries, including parity, geometrical, and permutation symmetries, both for $n=3$~\cite{WF05a,WF05b,YF09b,FSS+14,SA17b}, and for systems with $n>3$~\cite{Cab63,Jar05}.

It is also worth noting that, one may construct an overall reciprocal scattering matrix with nonreciprocal coupling. We first define the overall scattering matrix $S \in \mathbb{C}^{N \times N}$:
\begin{equation}
\label{eq:spq}
\mathbf{s_-}=S\mathbf{s_+}.
\end{equation}
Applying Equations~\ref{eq:cmt1} and~\ref{eq:cmt2} into Equation~\ref{eq:spq}, the scattering matrix can be written as:
\begin{equation}
\label{eq:smx}
S = C + D[j(\omega-\Omega)+\Gamma]^{-1}K^T.
\end{equation}
From Equation~\ref{eq:smx}, one can see that, even when reciprocity is maximally violated at the in- and out-coupling, for example, when $K=-D$, $S$ could be symmetric and reciprocal. In other words, internal symmetry breaking does not necessarily lead to external symmetry breaking. It is also of interest to quantitatively study the degree of symmetry breaking and the relaxation of the constraints in Tables~\ref{energy} and~\ref{time}.

As a final remark, although the temporal coupled-mode theory is an approximate and phenomenological model, its accuracy is sufficient for many applications, for instance, in electromagnetics~\cite{JJWM08,WL24}. A simpler version of Equation~\ref{eq:mag} was discovered in the temporal coupled-mode theory~\cite{WYSF13}, and was later discovered that similar bounds are generally applicable in electromagnetics from first-principles~\cite{SA17a}.

In summary, we discover time-reversal symmetry bounds on both coupling strengths and phases, as well as the trade-off therein. Given the importance of phase in physics, we expect the results to be exploited in many scenarios. In addition to the bounds on the system parameters, the distribution of the parameters might not be uniform within the bounds. Further work could also include exploring deep and rich results in nonreciprocal systems or with energy conversation only in Table~\ref{energy}, in topological systems~\cite{OPA+19,WDWF21,GLXF23}, in generalized, extended, or other versions of temporal coupled-mode theory~\cite{BD15,QDX+22,OMA24}, and beyond the temporal coupled-mode theory~\cite{DZ19}.

\appendix
\section*{Appendix: Complementary Minors in Unitary Matrix}

Complementary minors of a unitary matrix have equal magnitudes and complementary phases that add up to the phase of the determinant of the unitary matrix.

We consider a unitary matrix $\mathrm{U} \in U(n)$:
\begin{equation}
U = \left(
\begin{array}{cc}
A & B \\
C & D
\end{array}
\right).
\end{equation}
It suffices to prove for leading principal minors, as the results can be generalized to minors by permutation~\cite{Cot74,Mey00}. 

\begin{theorem}
\label{thm1}
Complementary minors of a unitary matrix have equal magnitude.
\end{theorem}

\begin{proof}
By unitarity, $A^\dagger A = I - B^\dagger B$ and $DD^\dagger  = I - BB^\dagger$. Since the determinant of sum of identity plus a matrix can be expressed by sums of principal minors of that matrix~\cite{Cos06,RS87}, and all the symmetric functions of $B^\dagger B$ and $BB^\dagger$ are equal to sums of respective principal minors~\cite{Mey00}, we can conclude that $\det(A^\dagger A) = \det(DD^\dagger)$ and  $\vert\det(A)\vert=\vert\det(D)\vert \leqslant \vert \det(U) \vert$.
\end{proof}

The statement can be generalized to non-square complementary submatrices as they share an identical nontrivial spectrum. The result can be further generalized with interlacing eigenvalues.

Although we can define the phase of a matrix~\cite{WCKQ20}, we restrict our attention to the phase of minors.

\begin{theorem}
Phases of complementary minors of a unitary matrix add up to the phase of the determinant of the unitary matrix.
\end{theorem}

\begin{proof}
For well-defined phases, $A$ and $D$ need to be invertible. Using Schur's complement~\cite{Mey00}, $\det(U) = \det(A-BD^{-1}C) \det(D) = \det(A) \det(I-A^{-1}BD^{-1}C) \det(D)$. By unitarity, and using the determinant of sum of matrices again, we can deduce that $\det(I-A^{-1}BD^{-1}C)$ is real. Hence $\angle\det(A) + \angle\det(D) =\angle\det(U)$ when the phase angle is properly defined.
\end{proof}

\printbibliography

@article{SF16,
	Author = {Yu Shi and Shanhui Fan},
	Journal = {Appl. Phys. Lett.},
	Volume = {108},
	Pages = {021110},
	Title = {Dynamic Non-Reciprocal Meta-Surfaces with Arbitrary Phase Reconfigurability Based on Photonic Transition in Meta-Atoms},
	Year = {2016}}

@article{SA17a,
	Author = {Dimitrios L. Sounas and Andrea Al{\`u}},
	Journal = {Phys. Rev. Lett.},
	Volume = {118},
	Pages = {154302},
	Title = {Time-Reversal Symmetry Bounds on the Electromagnetic Response of Asymmetric Structures},
	Year = {2017}}

@article{AYF17,
	Author = {Sid Assawaworrarit and Xiaofang Yu and Shanhui Fan},
	Journal = {Nature},
	Volume = {546},
	Pages = {387-390},
	Title = {Robust Wireless Power Transfer Using a Nonlinear Parity-Time Symmetric Circuit},
	Year = {2017}}

@article{SA17b,
	Author = {Dimitrios Sounas and Andrea Al{\`u}},
	Journal = {Nature Photon.},
	Volume = {11},
	Pages = {774-783},
	Title = {Non-Reciprocal Photonics Based on Time Modulation},
	Year = {2017}}

@article{YCH+15,
	Author = {Yi Yu and Yaohui Chen and Hao Hu and Weiqi Xue and Kresten Yvind and Jesper Mork},
	Journal = {Laser Photon. Rev.},
	Number = {2},
	Pages = {241-247},
	Title = {Nonreciprocal Transmission in a Nonlinear Photonic-Crystal {F}ano Structure with Broken Symmetry},
	Volume = {9},
	Year = {2015}}

@article{ESSA14,
	Author = {Nicholas A. Estep and Dimitrios L. Sounas and Jason Soric and Andrea Al{\`u}},
	Journal = {Nature Phys.},
	Pages = {923-927},
	Title = {Magnetic-Free Non-Reciprocity and Isolation Based on Parametrically Modulated Coupled-Resonator Loops},
	Volume = {10},
	Year = {2014}}

@article{MCH14,
	Author = {J. Mork and Y. Chen and M. Heuck},
	Journal = {Phys. Rev. Lett.},
	Pages = {163901},
	Title = {Photonic Crystal {F}ano Laser: Terahertz Modulation and Ultrashort Pulse Generation},
	Volume = {113},
	Year = {2014}}

@book{Sie86,
	Author = {Anthony E. Siegman},
	Publisher = {University Science Books},
	Title = {Lasers},
	Year = {1986}}

@article{XC15,
	Author = {Jing Xu and Yuntian Chen},
	Journal = {Opt. Express},
	Number = {17},
	Pages = {22619-22627},
	Title = {General Coupled Mode Theory in Non-{H}ermitian Waveguides},
	Volume = {23},
	Year = {2015}}

@article{JQW+15,
	Author = {Yiming Jia and Min Qiu and Hui Wu and Yi Cui and Shanhui Fan and Zhichao Ruan},
	Journal = {Nano Lett.},
	Number = {8},
	Pages = {5513-5518},
	Title = {Theory of Half-Space Light Absorption Enhancement for Leaky Mode Resonant Nanowires},
	Volume = {15},
	Year = {2015}}

@article{ZMZ12,
	Author = {Yanbing Zhang and Ting Mei and Dao Hua Zhang},
	Journal = {Appl. Opt.},
	Number = {4},
	Pages = {504-508},
	Title = {Temporal Coupled-Mode Theory of Ring-Bus-Ring {M}ach-{Z}ehnder Interferometer},
	Volume = {51},
	Year = {2012}}

@article{LLG+11,
	Author = {Hua Lu and Xueming Liu and Yongkang Gong and Dong Mao and Leiran Wang},
	Journal = {Opt. Express},
	Number = {14},
	Pages = {12885-12890},
	Title = {Enhancement of Transmission Efficiency of Nanoplasmonic Wavelength Demultiplexer Based on Channel Drop Filters and Reflection Nanocavities},
	Volume = {19},
	Year = {2011}}

@article{LAGL10,
	Author = {Thomas Lepetit and Eric Akmansoy and Jean-Pierre Ganne and Jean-Michel Lourtioz},
	Journal = {Phys. Rev. B},
	Pages = {195307},
	Title = {Resonance Continuum Coupling in High-Permittivity Dielectric Metamaterials},
	Volume = {82},
	Year = {2010}}

@article{RF12,
	Author = {Zhichao Ruan and Shanhui Fan},
	Journal = {Phys. Rev. A},
	Pages = {043828},
	Title = {Temporal Coupled-Mode Theory for Light Scattering by an Arbitrarily Shaped Object Supporting a Single Resonance},
	Volume = {85},
	Year = {2012}}

@article{SKM09,
	Author = {Hahn Young Song and Sangin Kim and Robert Magnusson},
	Journal = {Opt. Express},
	Number = {26},
	Pages = {23544-23555},
	Title = {Tunable Guided-Mode Resonances in Coupled Gratings},
	Volume = {17},
	Year = {2009}}

@article{ZSO+13,
	Author = {Linxiao Zhu and Sunil Sandhu and Clayton Otey and Shanhui Fan and Michael B. Sinclair and Ting Shan Luk},
	Journal = {Appl. Phys. Lett.},
	Pages = {103104},
	Title = {Temporal Coupled Mode Theory for Thermal Emission from a Single Thermal Emitter Supporting either a Single Mode or an Orthogonal Set of Modes},
	Volume = {102},
	Year = {2013}}

@article{FYLF11,
	Author = {Kejie Fang and Zongfu Yu and Victor Liu and Shanhui Fan},
	Journal = {Opt. Lett.},
	Number = {21},
	Pages = {4254-4256},
	Title = {Ultracompact Nonreciprocal Optical Isolator Based on Guided Resonance in a Magneto-Optical Photonic Crystal Slab},
	Volume = {36},
	Year = {2011}}

@phdthesis{Wang15,
	Author = {Ken Xingze Wang},
	title = {Photon Management in Solar Cells},
	school = {Stanford University},
	year = 2015,
}

@phdthesis{Dana18,
	Author = {Brenda Dana},
	title = {Advances in Theory and Applications of Coupled Mode Theory},
	school = {Tel Aviv University},
	year = 2018,
}

@article{WYL+14a,
	Author = {Ken Xingze Wang and Zongfu Yu and Victor Liu and Aaswath Raman and Yi Cui and Shanhui Fan},
	Journal = {Energy Environ. Sci.},
	Pages = {2725-2738},
	Title = {Light Trapping in Photonic Crystals},
	Volume = {7},
	Year = {2014}}

@inbook{Fan08,
	Author = {Shanhui Fan},
	Chapter = {12 - Photonic Crystal Theory: Temporal Coupled-Mode Formalism},
	Edition = {5th},
	Editor = {Ivan P. Kaminow and Tingye Li and Alan E. Willner},
	Pages = {431-454},
	Publisher = {Academic Press},
	Series = {Optics and Photonics},
	Title = {Optical Fiber Telecommunications},
	Volume = {A: Components and Subsystems},
	Year = {2008}}

@article{WF05a,
	Author = {Zheng Wang and Shanhui Fan},
	Journal = {Opt. Lett.},
	Pages = {1989-1991},
	Title = {Optical Circulators in Two-Dimensional Magneto-Optical Photonic Crystals},
	Volume = {30},
	Year = {2005}}

@article{WF05b,
	Author = {Z. Wang and S. Fan},
	Journal = {Appl. Phys. B},
	Pages = {369-375},
	Title = {Magneto-Optical Defects in Two-Dimensional Photonic Crystals},
	Volume = {81},
	Year = {2005}}

@book{Mey00,
	Author = {Carl D. Meyer},
	Publisher = {SIAM},
        Year = {2000},
	Title = {Matrix Analysis and Applied Linear Algebra},
	}

@unpublished{Cos06,
	Author = {R. S. Costas-Santos},
	Note = {arXiv:0612464},
        Year = {2006},
	Title = {On the Deterimant of a Sum of Matrices}}

@article{VYR+12,
	Author = {Lieven Verslegers and Zongfu Yu and Zhichao Ruan and Peter B. Catrysse and Shanhui Fan},
	Journal = {Phys. Rev. Lett.},
	Pages = {083902},
	Title = {From Electromagnetically Induced Transparency to Superscattering with a Single Structure: A Coupled-Mode Theory for Doubly Resonant Structures},
	Volume = {108},
	Year = {2012}}

@book{BD02,
	Author = {Anatoly A. Barybin and Victor A. Dmitriev},
	Publisher = {Rinton Press},
	Title = {Modern Electrodynamics and Coupled-Mode Theory: Application to Guided-Wave Optics},
	Year = {2002}}

@article{FSS+14,
	Author = {Romain Fleury and Dimitrios L. Sounas and Caleb F. Sieck and Michael R. Haberman and Andrea Al{\`u}},
	Journal = {Science},
	Number = {6170},
	Pages = {516-519},
	Title = {Sound Isolation and Giant Linear Nonreciprocity in a Compact Acoustic Circulator},
	Volume = {343},
	Year = {2014}}

@article{YF09b,
	Author = {Zongfu Yu and Shanhui Fan},
	Journal = {Nature Photon.},
	Pages = {91},
	Title = {Complete Optical Isolation Created by Indirect Interband Photonic Transitions},
	Volume = {3},
	Year = {2009}}

@article{HZL+13,
	Author = {Chia Wei Hsu and Bo Zhen and Jeongwon Lee and Song-Liang Chua and Steven G. Johnson and John D. Joannopoulos and Marin Solja$\mathrm{\check{c}}$i\'c},
	Journal = {Nature},
	Pages = {188},
	Title = {Observation of Trapped Light within the Radiation Continuum},
	Volume = {499},
	Year = {2013}}

@article{JBS+02,
	Author = {Steven G. Johnson and Peter Bienstman and M. A. Skorobogatiy and Mihai Ibanescu and Elefterios Lidorikis and J. D. Joannopoulos},
	Journal = {Phys. Rev. E},
	Pages = {066608},
	Title = {Adiabatic Theorem and Continuous Coupled-Mode Theory for Efficient Taper Transitions in Photonic Crystals},
	Volume = {66},
	Year = {2002}}

@article{YRF10b,
	Author = {Zongfu Yu and Aaswath Raman and Shanhui Fan},
	Journal = {Proc. Natl. Acad. Sci. U.S.A.},
	Number = {41},
	Pages = {17491-17496},
	Title = {Fundamental Limit of Light Trapping in Solar Cells},
	Volume = {107},
	Year = {2010}}

@article{SWF04,
	Author = {Wonjoo Suh and Zheng Wang and Shanhui Fan},
	Journal = {IEEE J. Quant. Electron.},
	Number = {10},
	Pages = {1511},
	Title = {Temporal Coupled-Mode Theory and the Presence of Non-Orthogonal Modes in Lossless Multimode Cavities},
	Volume = {40},
	Year = {2004}}

@book{Hau83,
	Author = {Hermann A. Haus},
	Publisher = {Prentice Hall},
	Title = {Waves and Fields in Optoelectronics},
	Year = {1983}}

@article{HH91,
	Author = {Hermann A. Haus and Weiping Huang},
	Journal = {Proc. IEEE},
	Number = {10},
	Pages = {1505},
	Title = {Coupled-Mode Theory},
	Volume = {79},
	Year = {1991}}

@article{FSJ03,
	Author = {Shanhui Fan and Wonjoo Suh and J. D. Joannopoulos},
	Journal = {J. Opt. Soc. Am. A},
	Month = {March},
	Number = {3},
	Pages = {569-572},
	Title = {Temporal Coupled-Mode Theory for the {F}ano Resonance in Optical Resonators},
	Volume = {20},
	Year = {2003}}

@article{WYSF13,
	Author = {Ken Xingze Wang and Zongfu Yu and Sunil Sandhu and Shanhui Fan},
	Journal = {Opt. Lett.},
	Month = {January},
	Number = {2},
	Pages = {100-102},
	Title = {Fundamental Bounds on Decay Rates in Asymmetric Single-Mode Optical Resonators},
	Volume = {38},
	Year = {2013}}

@book{JJWM08,
	Author = {John D. Joannopoulos and Steven G. Johnson and Joshua N. Winn and Robert D. Meade},
	Edition = {2nd},
	Publisher = {Princeton University Press},
	Title = {Photonic Crystals: Molding the Flow of Light},
	Year = {2008}}

@book{KSM18,
	Author = {Eugene Kamenetskii and Almas Sadreev and Andrey Miroshnichenko},
	Publisher = {Springer},
	Title = {Fano Resonances in Optics and Microwaves},
	Year = {2018}}

@article{LHZ+23,
	Author = {Shanshan Liu and Sibo Huang and Zhiling Zhou and Pei Qian and Bin Jia and Hua Ding and Nengyin Wang and Yong Li and Jie Chen},
	Journal = {Phys. Rev. Applied},
	Pages = {044075},
	Title = {Observation of acoustic Friedrich-Wintgen bound state in the continuum with bridging near-field coupling},
	Volume = {20},
	Year = {2023}}

@article{Pile12,
	Author = {David Pile},
	Journal = {Nature Photonics},
	Pages = {412},
	Title = {Time-dependence},
	Volume = {6},
	Year = {2012}}

@article{CTK24,
	Author = {Thomas Christopoulos and Odysseas Tsilipakos and Emmanouil E. Kriezis},
	Journal = {J. Appl. Phys.},
	Pages = {011101},
	Title = {Temporal coupled-mode theory in nonlinear resonant photonics: From basic principles to contemporary systems with 2D materials, dispersion, loss, and gain},
	Volume = {136},
	Year = {2024}}

@article{DPM18,
	Author = {Victor Dmitriev and Gianni Portela and Leno Martins},
	Journal = {IEEE Transactions on Microwave Theory and Techniques},
	Pages = {1165-1171},
	Title = {Temporal Coupled-Mode Theory of Electromagnetic Components Described by Magnetic Groups of Symmetry},
	Volume = {66},
	Year = {2018}}

@article{QDX+22,
	Author = {Meibao Qin and Junyi Duan and Shuyuan Xiao and Wenxing Liu and Tianbao Yu and Tongbiao Wang and Qinghua Liao},
	Journal = {Phys. Rev. B},
	Pages = {195425},
	Title = {Manipulating strong coupling between exciton and quasibound states in the continuum resonance},
	Volume = {105},
	Year = {2022}}

@article{BD15,
	Author = {Dmitry A. Bykov and Leonid L. Doskolovich},
	Journal = {Optics Express},
	Pages = {19234-19241},
	Title = {Spatiotemporal coupled-mode theory of guided-mode resonant gratings},
	Volume = {23},
	Year = {2015}}

@article{OMA24,
	Author = {Adam Overvig and Sander A. Mann and Andrea Alù},
	Journal = {Light Sci. Appl.},
	Pages = {28},
	Title = {Spatio-temporal coupled mode theory for nonlocal metasurfaces},
	Volume = {13},
	Year = {2024}}

@article{WL24,
	Author = {Tong Wu and Philippe Lalanne},
	Journal = {Optics Express},
	Pages = {20904-20914},
	Title = {Exact Maxwell evolution equation of resonator dynamics: temporal coupled-mode theory revisited},
	Volume = {32},
	Year = {2024}}

@article{WDWF21,
	Author = {Kai Wang and Avik Dutt and Charles C. Wojcik and Shanhui Fan},
	Journal = {Nature},
	Pages = {59-64},
	Title = {Topological complex-energy braiding of non-Hermitian bands},
	Volume = {598},
	Year = {2021}}

@article{SM22,
	Author = {Olivier Sigwarth and Christian Miniatura},
	Journal = {AAPPS Bulletin},
	Pages = {23},
	Title = {Time reversal and reciprocity},
	Volume = {32},
	Year = {2022}}

@article{ZGF19,
	Author = {Zhexin Zhao and Cheng Guo and Shanhui Fan},
	Journal = {Phys. Rev. A},
	Pages = {033839},
	Title = {Connection of temporal coupled-mode-theory formalisms for a resonant optical system and its time-reversal conjugate},
	Volume = {99},
	Year = {2019}}

@book{Ach04,
	Author = {J. D. Achenbach},
	Publisher = {Cambridge University Press},
	Title = {Reciprocity in Elastodynamics},
	Year = {2004}}

@article{Jar05,
	Author = {C. Jarlskog},
	Journal = {J. Math. Phys.},
	Pages = {103508},
	Title = {A recursive parametrization of unitary matrices},
	Volume = {46},
	Year = {2005}}

@article{Bro88,
	Author = {J. B. Bronzan},
	Journal = {Phys. Rev. D},
	Pages = {38},
	Title = {Parametrization of SU(3)},
	Volume = {38},
	Year = {1988}}

@article{Cab63,
	Author = {Nicola Cabibbo},
	Journal = {Phys. Rev. Lett.},
	Pages = {531},
	Title = {Unitary Symmetry and Leptonic Decays},
	Volume = {10},
	Year = {1963}}

@article{RS87,
	Author = {Christophe Reutenauer and Marcel-Paul Schützenberger},
	Journal = {Lett. Math. Phys.},
	Pages = {299-302},
	Title = {A formula for the determinant of a sum of matrices},
	Volume = {13},
	Year = {1987}}

@article{Wang18,
	Author = {Ken Xingze Wang},
	Journal = {Opt. Lett.},
	Pages = {5623},
	Title = {Time-reversal symmetry in temporal coupled-mode theory and nonreciprocal device applications},
	Volume = {43},
	Year = {2018}}

@article{Cot74,
	Author = {Richard W. Cottle},
	Journal = {Linear Algebra and its Applications},
	Pages = {189-211},
	Title = {Manifestations of the Schur complement},
	Volume = {8},
	Year = {1974}}

@article{GF22,
	Author = {Cheng Guo and Shanhui Fan},
	Journal = {Phys. Rev. Lett.},
	Pages = {256101},
	Title = {Reciprocity Constraints on Reflection},
	Volume = {128},
	Year = {2022}}

@article{GLXF23,
	Author = {Cheng Guo and Jiazheng Li and Meng Xiao and Shanhui Fan},
	Journal = {Phys. Rev. B},
	Pages = {155418},
	Title = {Singular topology of scattering matrices},
	Volume = {108},
	Year = {2023}}

@article{OPA+19,
	Author = {Tomoki Ozawa and Hannah M. Price and Alberto Amo and Nathan Goldman and Mohammad Hafezi and Ling Lu and Mikael C. Rechtsman and David Schuster and Jonathan Simon and Oded Zilberberg and Iacopo Carusotto},
	Journal = {Review of Modern Physics},
	Pages = {015006},
	Title = {Topological photonics},
	Volume = {91},
	Year = {2019}}

@article{GZF22,
	Author = {Cheng Guo and Zhexin Zhao and Shanhui Fan},
	Journal = {Phys. Rev. A},
	Pages = {023509},
	Title = {Internal transformations and internal symmetries in linear photonic systems},
	Volume = {105},
	Year = {2022}}

@article{GLY+23,
	Author = {Xin Gu and Xing Liu and Xiao-Fei Yan and Wen-Juan Du and Qi Lin and Ling-Ling Wang and Gui-Dong Liu},
	Journal = {Optics Express},
	Pages = {4691-4700},
	Title = {Polaritonic coherent perfect absorption based on self-hybridization of a quasi-bound state in the continuum and exciton},
	Volume = {31},
	Year = {2023}}

@article{WCKQ20,
	Author = {Dan Wang and Wei Chen and Sei Zhen Khong and Li Qiu},
	Journal = {Linear Algebra and its Applications},
	Pages = {152-179},
	Title = {On the phases of a complex matrix},
	Volume = {593},
	Year = {2020}}

@book{DZ19,
	Author = {Semyon Dyatlov and Maciej Zworski},
	Publisher = {American Mathematical Society},
	Title = {Mathematical Theory of Scattering Resonances},
	Year = {2019}}
\end{document}